\documentclass[11pt]{amsart}

\usepackage{amsmath,amssymb,amscd,amsthm,amsfonts}
\usepackage{graphicx,subfigure}
\usepackage{hyperref}
\usepackage{pstricks}
\usepackage{dsfont}
\usepackage{pifont}

\newtheorem{theorem}{Theorem}[section]

\newtheorem*{theoremp}{Theorem}
\newtheorem{lemma}[theorem]{Lemma}

\newtheorem{definition}[theorem]{Definition}
\usepackage{tikz}

\newcommand{\rr}{\mathds{R}}

\newcommand{\ff}{\mathcal{F}}

\DeclareMathOperator{\conv}{conv}

\DeclareMathOperator{\poly}{poly}
\DeclareMathOperator{\deep}{depth}

\title{Algorithms for Tverberg's theorem via centerpoint theorems}

\author{D.~Rolnick}
\thanks{Massachusetts Institute of Technology, Cambridge, MA. \texttt{drolnick@mit.edu}}
\author{P.~Sober\'on}
\thanks{Northeastern University, Boston, MA. \texttt{p.soberonbravo@northeastern.edu}}

\begin{document}

\begin{abstract}
We obtain algorithms for computing Tverberg partitions based on centerpoint approximations.  This applies to a wide range of convexity spaces, from the classic Euclidean setting to geodetic convexity in graphs.  In the Euclidean setting, we present probabilistic algorithms which are weakly polynomial in the number of points and the dimension.  For geodetic convexity in graphs, we obtain deterministic algorithms for cactus graphs and show that the general problem of finding the Radon number is NP-hard.
\end{abstract}

\maketitle


\section{Introduction}

Radon's lemma and Tverberg's theorem are central results in combinatorial geometry \cite{Radon:1921vh, Tverberg:1966tb}.  These theorems describe the size at which a point set can be dissected into overlapping convex hulls.

\begin{theoremp}[Tverberg's theorem]
Given a set of $(k-1)(d+1)+1$ points in $\rr^d$, there is a partition of the set into $k$ parts such that the convex hulls of the parts intersect. Furthermore, this bound is optimal.
\end{theoremp}

The case $k=2$ is known as Radon's lemma or Radon's theorem.  There are many generalizations and extensions of Tverberg's theorem, such as colorful \cite{Barany:1992tx, Blagojevic:2011vh, Blagojevic:2009wya}, topological \cite{Barany:1981vh,  frick2015, Volovikov:1996up}, and quantitative versions \cite{deloera2015}.  There are also connections to Helly's and Carath\'eodory's theorems. (See the survey \cite{Eck93} for a more in-depth presentation of the links between classic theorems in combinatorial geometry.)  Another gem in discrete geometry is Rado's centerpoint theorem \cite{rado1946theorem}.

\begin{theoremp}[Centerpoint theorem]
Let $d$ be a positive integer.  Given any finite set $X \subset \rr^d$, there is a point $p$ such that every closed half-space that contains $p$ also contains at least $\frac{|X|}{d+1}$ points of $X$.
\end{theoremp}

  Tverberg's theorem can be thought of as a discrete strengthening of the centerpoint theorem.  Indeed, if $n = (k-1)(d+1)+1$, then $k = \left\lceil \frac{n}{d+1} \right\rceil$.  If $p$ is the point of intersection of a Tverberg partition, any half-space that contains $p$ also has at least one point from each part of the partition.  This implies that the point of intersection of a Tverberg partition is a centerpoint.  The problem of finding Tverberg partitions efficiently was motivated by the problem of finding centerpoints efficiently (see, for instance, \cite{Miller:2009bg}).
  
  Given a set $S$ of $n$ points in $\rr^d$, we say a point $p$ is an $\alpha$-centerpoint of $S$ if every closed halfspace containing $p$ has at least $\alpha n$ points of $S$.  Notice that any point in the intersection of a Tverberg partition $A_1, A_2, \ldots, A_k$ is immediately an $(k/n)$-centerpoint.  As Tverberg partitions can be checked for correctness, this provides a robust way of finding centerpoints.

  The aim of this paper is to show that the other direction is also interesting from an algorithmic perspective.  In other words, in order to compute (or approximate) Tverberg partitions, we can do so by computing centerpoints.  Moreover, since centerpoints can be easily generalized, the algorithms we present extend to various convexity spaces.
  
For the classic notion of convexity in $\rr^d$, it is not known if there exists an efficient algorithm for constructing Tverberg partitions under the assumptions of Tverberg's theorem. The case $k=2$ (finding a Radon partition) is, however, simple; this problem reduces to identifying a linear dependence. A similar technique can also be applied to variations on Radon partitions, such as the colorful version of Radon's lemma \cite{Soberon:2013fr}, giving an algorithmic proof. The interesting challenge lies in Tverberg-type results with $k>2$.

Obtaining optimal Tverberg partitions efficiently is out of reach for current algorithms.  However, if one is willing to pay the price of decreasing the value of $k$ slightly, better results can be obtained.  For instance, there is a deterministic algorithm by Miller and Sheehy that solves the problem for $k = \left\lceil \frac{n}{(d+1)^2}\right\rceil$ in $n^{O(\log d)}$ time \cite{Miller:2009bg}.  For faster running times in terms of $n$, there is an algorithm by Mulzer and Werner that solves the problem for $k = \left\lceil \frac{n}{4(d+1)^3}\right\rceil$ in $d^{O(\log d)}n$ time \cite{MW13}.  Note that even though this last algorithm is linear in $n$, it is superpolynomial in the dimension.  Related algorithms can be found for variations of Tverberg such as the version with tolerance \cite{MS14, Soberon:2012er}.  In this variation, given a positive integer $r$, the goal is to find a Tverberg partition with the property that even after removing any $r$ points the convex hulls of what is left on each part still intersect.

In Section \ref{sec:classic}, we present our algorithms for convexity in $\rr^d$.  Our algorithms are not deterministic. They carry a probability $\varepsilon$ of failure, which may be fixed in advance, as they depend heavily on the computation of approximated centerpoints by Clarkson et al.~\cite{Clarkson:1996ixba}.  For $k \sim n/d^2$ we present an algorithm that is weakly polynomial in $n$, but exponential in the dimension.  For $k \sim n/d^3$ we present an algorithm that is weakly polynomial in $n$, $d$ and $\log (1/\varepsilon)$.  This would be the first algorithm that is weakly polynomial in all variables.

\begin{theorem}\label{theorem-centerpoint}
 Let $k = \left\lceil \frac{n}{d(d+1)^2}\right\rceil$.  There is an algorithm that, given $n$ points in $\rr^d$, finds an $k$-Tverberg partition in time weakly polynomial in $n, d, \log(1/\varepsilon)$ with error probability $\varepsilon$.
\end{theorem}

\begin{theorem}\label{theorem-nonpolynomial-centerpoint}
Let $d, \lambda$ be fixed, where $d$ is a positive integer and $0<\lambda<\frac{1}{d+1}$, and $k = \left\lceil n\left(\frac{1}{d+1}-\lambda \right)\frac{1}{d}\right\rceil$.  Then, there is an algorithm that, given a set of $n$ points in $\rr^d$, finds an $k$-Tverberg partition in weakly polynomial time in $n$, $O({n^4} \log(1/{\varepsilon}))$, with error probability $\varepsilon$.
\end{theorem}

In Theorem \ref{theorem-nonpolynomial-centerpoint}, if we allow $d, \lambda$ to vary, we need an additional factor of $(d/{\lambda})^{O(d)}$.   As we mentioned before, the same ideas can be used in more general settings.  For example, we obtain algorithms for the integer version of Tverberg's theorem, which was first described in \cite{Eck68}:

\begin{theoremp}
Given $k,d$ positive integers, there is an integer $T=T(k,d)$ such that for any set of $T$ points in $\rr^d$ with integer coordinates, there is a partition of the set into $k$ parts such that the intersection of the convex hulls of the parts contains a point with integer coordinates.
\end{theoremp}

The exact values for $T(k,d)$ remain unknown.  Even the case $k=2$ is not completely solved, as the best current bounds are $5\cdot 2^{d-2}+1 \le T(2,d) \le d\cdot (2^d-1)+3$ and $T(k,d) \le (k-1)d\cdot 2^d+1$ \cite{deloera2015, Onn:1991er}.  Our algorithms yield the following:

\begin{theorem}\label{theorem-integer-tverberg}
Let $d, \lambda, \varepsilon$ be fixed and that $0<\lambda<{2^{-(d+2)}}$.  Then, there is an algorithm that is weakly polynomial in $n$ that, for any set of $n$ integer points in $\rr^d$, gives (i) a partition of them into $\left\lceil \frac{n}{d} \left(\frac{1}{2^{d}}-\lambda \right) \right\rceil$ parts and (ii) an integer point $z$ in the convex hull of every part, with error probability $\varepsilon$.
\end{theorem}

As mentioned before, the same type of ideas work for more general convexity spaces.  We say a pair $(X, \conv)$ is a \textit{convexity space} \cite{van1993theory} if $X$ is a set and $\conv : 2^X \to 2^X$ is a function that satisfies
\begin{itemize}
	\item $A \subset \conv (A)$ for all $A \subset X$,
	\item if $A \subset B \subset X$, then $\conv(A) \subset \conv (B)$, 
	\item $\conv (A) = \conv (\conv (A))$ for all $ A \subset X$, and
	\item If $A_1 \subset A_2 \subset \ldots$ is an infinite sequence of nested sets, then $A=\cup_{i=1}^{\infty}\conv (A_i)$ is convex, i.e. $A = \conv (A)$.
\end{itemize}

For instance, if $X = \rr^d$ and $\conv(\cdot)$ is the classic convex hull, we have a convexity space.  However, the definition above is relaxed enough to capture many purely combinatorial instances.  Given a connected graph $G$ and two vertices $u,v \in V(G)$, we can define the segment $[u,v]$ as the union of all shortest paths from $u$ to $v$.  Then, for any $A \subset V(G)$, we can take $\conv(A) = \cup_{u,v \in A}[u,v]$, yielding what is known as a \emph{geodetic convexity space}.  Tverberg-type results are interesting only for convexity spaces in which $\conv (\emptyset) = \emptyset$.

\begin{definition}
Given a convexity space $(X, \conv)$, the $k$th \emph{Radon number} (if it exists) is the smallest integer $r_k$ such that for any $r_k$ points in $X$, there is a $k$-partition such that the convex hulls of all parts intersect.
\end{definition}

Tverberg and Radon-type results in general convexity spaces are much more enigmatic than in the setting of $\rr^d$.  A classic conjecture by Eckhoff, known as the partition conjecture \cite{Eckhoff:2000jwa}, states that if $r_2$ exists, then $r_k$ exists and is at most $(r_2-1)(k-1)+1$.  Thus, if this conjecture holds, Tverberg's theorem is a purely combinatorial consequence of Radon's lemma.  Even though there have been convexity spaces proposed by Bukh showing that the partition conjecture does not hold in general \cite{bukh}, it is interesting to know for which convexity spaces it does hold.

It was shown by Duchet that the partition conjecture on finite convexity spaces in fact reduces to the case of geodetic convexity spaces on graphs \cite{duchet1998}. Jamison corroborated the partition conjecture for certain graphs \cite{Jamison:1981wz}. However, geodetic convexity spaces for general graphs are little understood, and even the Radon numbers remain unknown for many graphs. Results for grid graphs were recently presented by Dourado et al.~in \cite{dourado2013,dourado2013polynomial}.

In general convexity spaces, we need to modify our definition of centerpoints accordingly.  The definition below is equivalent to that of an $\alpha$-centerpoint in the usual convexity in $\rr^d$.

\begin{definition}
	Let $(X, \conv)$ be a convexity space and $Y \subset X$ be a finite subset.  Then, we say that $p$ is an $\alpha$-centerpoint of $Y$ if and only if for every subset $Z \subset Y$ such that $|Z| > (1-\alpha)|Y|$, we have $p \in \conv (Z)$.
\label{def:centerpoint}
\end{definition}

In Section \ref{sec:geodetic-positive}, we present algorithms for finding Tverberg partitions in geodetic convexity spaces induced by certain graphs. These (deterministic) algorithms rely on finding centerpoints and using them to obtain Tverberg partitions.  Given a graph $G$, we say it is a \emph{cactus graph} (equivalently, it is \emph{2-separable}) if it is connected and for any two vertices $u,v$ there are at most two edges such that removing them disconnects $u$ from $v$.

\begin{theorem}\label{theorem-for-trees}
Let $G$ be a tree on $n$ vertices and $U \subset V$ be a subset of $2k$ vertices.  Then, there exists a partition of $U$ into $k$ sets such that their convex hulls intersect.  Moreover this partition may be found in time at most linear in $n$.
\end{theorem}

\begin{theorem}\label{theorem-for-2-sep}
Let $G$ be a cactus graph on $n$ vertices, and $U$ a set of $4k-2$ vertices of $G$.  Then, there exists a partition of $U$ into $k$ parts such that their convex hulls intersect.  Moreover this partition may be found in time at most quadratic in $n$.
\end{theorem}


We conclude in Section \ref{sec:geodetic-negative} by presenting results on the hardness of finding Radon partitions for the geodetic convexity spaces induced by general graphs.  We define the \emph{graph-Radon problem} as follows:\\
\textbf{Input:} A graph $G=(V,E)$ and subset $W\subset V$.\\
\textbf{Output:} A decision whether or not there exists a partition $(W^+,W^-)$ of $W$ such that $\conv(W^+)\cap \conv(W^-)$ is nonempty.

Likewise, we define the \emph{graph-Radon counting problem}:\\
\textbf{Input:} A graph $G=(V,E)$ and subset $W\subset V$.\\
\textbf{Output:} The number of partitions $(W^+,W^-)$ of $W$ such that $\conv(W^+)\cap \conv(W^-)$ is nonempty.

\begin{theorem}
\label{thm:onepartition}
The graph-Radon problem is NP-hard and the graph-Radon counting problem is \#P-hard.
\end{theorem}

As a corollary to this result, we present a novel proof that it is NP-hard to compute the Radon number of the geodetic convexity space induced by a general graph, proven recently by Coelho et al.~\cite{coelho2015}.


\section{Algorithms for classic Tverberg partitions}
\label{sec:classic}

We make use of Carath\'eodory's theorem in order to find Tverberg partitions.

\begin{theoremp}[Carath\'eodory's theorem \cite{Car07}]
Given a set $X \subset \rr^d$ and a point $z \in \conv (X)$, there is a subset $C$ of at most $d+1$ points of $X$ such that $z \in \conv(X)$.
\end{theoremp}

\begin{lemma}\label{lemma-caratheodory}
Given $X$ and $z$, we can find $C$ using linear programming in time weakly polynomial in $|X|$ and $d$.
\end{lemma}

\begin{proof}[Proof of Lemma \ref{lemma-caratheodory}]
Deciding whether a point $z$ is in the convex hull of a set $Y$ can be expressed by the following linear program, where the $\lambda_y$ represent variables and $y\in Y$ are $d$-dimensional constant vectors:
\begin{align*}
\sum_{y\in Y} \lambda_y y &= z\\
\sum_{y \in Y} \lambda_y &= 1 \\
\lambda_y &\ge 0\text{ for all $y\in Y$}.
\end{align*}

Linear programs can be solved in weakly polynomial time in the number $n$ of variables and the number $L$ of bits used to express the program. Specifically, Ye's interior point algorithm achieves a running time of $O(n^3L)$.  We apply this observation as follows: For each $x$ in $X$, we can check if $z$ is in $\conv(X\setminus\{x\})$. If we succeed in finding such a point, we remove $x$ from $X$ and continue. Carath\'edory's theorem guarantees that we will be left with at most $d+1$ points containing $z$ in their convex hull.

This procedure takes $\tilde{O}(n^4)$ time, with a tacit dependence on $L$; hence the algorithm is weakly polynomial. A slight improvement is achieved by the following procedure: Define a set of at most $d+1$ \emph{essential} vertices in $X$, which contain $z$ in their convex hull. If a subset of $X$ does not contain $z$ in its convex hull, then it must be missing at least one essential vertex. We can use this information to find the essential vertices by binary search: First eliminate half the vertices of $X$ and check if their convex hull contains $z$. If it does contain $z$, continue to eliminate vertices; otherwise, replace half the vertices that were eliminated and try again.  Finding each essential vertex in this fashion takes $\log n$ time, and there are at most $d+1$ of them. Therefore, we need $O(d\log n)$ applications of our LP-solver, giving total running time $\tilde{O}(n^3)$. We do not make an effort to optimize this running time further.
\end{proof}

\begin{lemma}\label{lemma-simple}
Let $\alpha$ be in $(0,1]$ and $d$ be a positive integer.  Suppose there is an algorithm that runs in time $f(\alpha, d, n)$ such that given any set $X$ of $n$ points in $\rr^d$ as input, it gives an $\alpha$-centerpoint $z \in \rr^d$ of $X$.  In other words, every halfspace containing $z$ contains at least $\alpha n$ points of $X$.  Assume as well that given a set $X$ of $n$ points in $\rr^d$ and $x \in \conv (X)$, there is an algorithm that runs in time $\beta  (n,d)$ and gives a set $C \subset X$ of at most $d+1$ points such that $x \in \conv (C)$.

Then, there is an algorithm that, given a set $X$ of $n$ points in $\rr^d$, finds a Tverberg partition of size $\left \lceil \frac{\alpha n}{d}\right\rceil$ in time $O\left( f(\alpha, d, n)+\beta(n,d) \cdot \alpha n\cdot \poly (d) \right)$.
\end{lemma}

\begin{proof}[Proof of Lemma \ref{lemma-simple}]
Given a set $X$ and a point $z$, denote by $\deep (z,X)$ the minimum number of points in any halfspace containing $z$. We begin by applying the first algorithm to find a point $z$ such that $\deep (z,X) \ge \alpha n$, using time $f(\alpha, d, n)$. Then, as long as $\deep (z,X) >0$, we can find a set $C$ of at most $d+1$ points such that $z \in \conv (C)$, using the second algorithm.  We may assume that if $|C|=d+1$, it is minimal, by checking if $z$ lies in the convex hull of any $d$-element subset of $C$.  Notice that for each closed halfspace $H^+$ such that its boundary hyperplane contains $z$, we have $H^+ \cap C \le d$.  Removing $C$ from $X$ reduces $\deep (z, X)$ by at most $d$.  Thus, we can repeat this process at least $\frac{\alpha n}{ d}$ times.  The family of sets $C$ obtained in each step gives us a Tverberg partition that intersects in $z$, as desired.
\end{proof}

Notice that if finding the point $z$ employs an algorithm with a probability $\varepsilon$ of failing, then the overall algorithm also has probability $\varepsilon$ of failing.

Theorems \ref{theorem-centerpoint} and \ref{theorem-nonpolynomial-centerpoint} follow from the method above, given the algorithms for centerpoints presented in Clarkson et al.~\cite{Clarkson:1996ixba}.  In order to obtain the results for the integer lattice, we generalize the methods of Clarkson et al.~using a notion of centerpoints with integer coordinates.  The existence of these can be settled by using Doignon's theorem, an integer counterpart to Helly's theorem which was later rediscovered twice with integer optimization in mind \cite{Bell:1977tm, Scarf:1977va}.

\begin{theorem}[Doignon's theorem \cite{Doignon:1973ht}]
Let $\mathcal{F}$ be a finite family of convex sets in $\rr^d$.  If every $2^d$ sets in $\mathcal{F}$ have an integer point in common, then so does the entire family $\ff$.
\end{theorem}

\begin{lemma}\label{lemma-doignon}
Given any finite set $S$ of integer points in $\rr^d$, there is an integer point $z$ such that every halfspace containing $z$ has at least $\left\lceil \frac{|S|}{2^d} \right\rceil$ points of $S$.
\end{lemma}

\begin{proof}[Proof of Lemma \ref{lemma-doignon}]
Consider the family
\[
\mathcal{F} = \left\{\conv(K) : K\subset S, \ |K|>\frac{(2^d-1) }{2^d}|S|\right\}.
\]
By construction, $\mathcal{F}$ satisfies the conditions of Doignon's theorem.  Let $z$ be an integer point in $\cap\mathcal{F}$.  If a halfspace containing $z$ had fewer than $\frac{|S|}{2^d}$ points of $S$, it would mean that $z$ would be separated by a hyperplane by a subset of $S$ of cardinality strictly greater than $\frac{(2^d-1) }{2^d}|S|$, contradicting the fact that it lies in $\cap \mathcal{F}$.
\end{proof}

In order to construct integer centerpoints in a manner analogous to \cite{Clarkson:1996ixba}, we recall the notion of $\lambda$-samples.

\begin{definition}
Given a set $S\subset \rr^d$ with $n$ points, we call $S' \subset S$ a $\lambda$-sample if for any halfspace $H$ with $|H \cap S| \ge 4\lambda n$, we also have
\[
\frac{|H \cap S'|}{|S'|} \le \frac{|H \cap S|}{n}+\lambda.
\]
\end{definition}

In particular, a $\beta$-center of a $\lambda$-sample of $S$ is immediately a $(\beta - \lambda)$-center of $S$, for $\beta \ge 4\lambda$.  In \cite{Clarkson:1996ixba}, it was shown that a random sample of sufficiently large size, but depending only on $d,\lambda, \varepsilon$, is a $\lambda$-sample with probability at least $1-\varepsilon$.

Thus, given a set of $n$ integer points and $0<\lambda \le 2^{-(d+2)}$, we can find an integer $\left(\frac{1}{2^{d}}-\lambda\right)$-center with probability at least $1-\varepsilon$ in the following manner: First, we find a $\lambda$-sample $S'$ of the set with probability $1-\varepsilon$.  Then, finding an integer $\frac{1}{2^d}$-center for $S'$ is an integer programming problem where the number of constraints is fixed (as $d, \lambda, \varepsilon$ are fixed) so it can be solved in polynomial time. We know that a solution exists by Doignon's theorem, and the solution must be a $\left(\frac{1}{2^d}-\lambda\right)$-center of the original set, as desired.


\section{Algorithms for geodetic Tverberg partitions}
\label{sec:geodetic-positive}

In this section we present algorithms for Tverberg partitions in the geodetic convexity spaces of certain graphs. The partition conjecture was already established for trees by Jamison \cite{Jamison:1981wz}; our proof of Theorem \ref{theorem-for-trees} gives the optimal Tverberg number and shows how such a partition may be found algorithmically.

\begin{proof}[Proof of Theorem \ref{theorem-for-trees}]
Let $p$ be an arbitrary vertex of the tree $G$.  Suppose that removing $p$ splits $G$ into $r$ connected components $G_1, G_2, \ldots, G_r$, and let $a_i = |G_i \cap U|$ for all $i$.  We first show that we can find $p$ such that $a_i \le k$ for all $i$.  If this is not the case, assume without loss of generality that $a_1 \ge k+1$.  Notice that there is a vertex $p'$ in $G_1$ which is adjacent to $p$, so we replace $p$ by $p'$.  By doing this, $G_2, G_3, \ldots , G_r$ and $\{p\}$ merge into one big component, but since they had altogether at most $k$ points of $U$, they do not cause a problem.  $G_1 \setminus \{p\}$ may be split into more than one component, but at most one of them has more than $k$ points of $S$.  If this is the case, call it $G_1'$.  Since $G_1'$ has fewer vertices than $G_1$, this process must end after at most $|G|$ iterations.

In order to find the partition, define $G_0 = \{p\}$.  It suffices to notice that if we take two vertices of $U$ in different components $G_i$, $G_j$, their convex hull must contain $p$.  For any partition of a set of $2k$ vertices in which no part has more than $k$ vertices, there is a way to pair up points such that no pair lies in the same part.  Therefore, we obtain the desired result.  The desired pairing can be also be found algorithmically.  Simply take the two parts with the largest number of vertices $G_1$, $G_2$, pick two arbitrary vertices $v_1 \in G_1, v_2 \in G_2$, make them an edge of the matching and set them aside.  Then, continue in the same way with the rest of the graph.  This process ensures that at every stage no component has more than half of the remaining vertices, so we can continue until the graph is empty.

Let us now consider the time required to find a vertex $p$ of the form desired. For each vertex $q$ of $G$ let $u(q)$ be the tuple corresponding to the numbers of elements of $U$ in the components of $G\setminus \{q\}$.  We can identify $u(q)$ recursively for every vertex $q\in G$, as follows. For $q$ a leaf of $G$, the tuple $u(q)$ is simply $(k)$.  Starting from leaves of $G$, we work inward, updating $u(q)$ when we reach either a branch point within $G$ or an element of $U$.  Since each vertex need be considered only once, the time required is linear; this procedure must identify a vertex $p$ of the form desired.
\end{proof}

In order to show that the value $2k$ is optimal, take any set of size $2k-1$ of leaves of a tree.  Since any partition of them into $k$ sets must include a singleton, there is no Tverberg partition into $k$ parts.  From this theorem, we may verify Eckhoff's partition conjecture in the case of geodetic convexity spaces on trees:
$$r_k = 2k \le 3(k-1)+1 = (r_2-1)(k-1)+1.$$

One must beware of the special case in which the Radon number of the tree is less than four.  However, this only happens with trees with exactly two leaves, reducing to the case of Radon on the real line.  It should be stressed that the algorithm presented here is similar to those in the previous section.  Indeed, our point $p$ is a centerpoint for $U$ in the sense of Definition \ref{def:centerpoint}, as the convex hull of any subset of $U$ of cardinality greater than $|U|/2$ contains $p$. The same idea is extended to another family of graphs in Theorem \ref{theorem-for-2-sep}. In this case instead of a single centerpoint, we look for a pair of vertices that in some sense serve jointly as centerpoints.

\begin{proof}[Proof of Theorem \ref{theorem-for-2-sep}]
If we could show that there are two vertices $x,y$ such that no connected component of $G\setminus \{x,y\}$ has more than $2k-1$ vertices of $U$, we would be done.  This is because in that case we could find a partition of $U$ into $2k-1$ pairs such that the convex hull of each pair contains at least one of $x,y$.  By the pigeonhole principle, this would imply that at least $k$ pairs intersects in either $x$ or $y$, yielding the partition we seek.

Notice that in a cactus graph, two different cycles cannot share more than one vertex.  Otherwise, we would be able to find two vertices with three edge-disjoint paths connecting them.

We construct an auxiliary graph $G'$ as follows.  First, we color the graph $G$ blue.  Then, consider any vertex $v$ that is in more than one cycle.  Notice that these cycles are pairwise edge-disjoint.  If $v$ is in $n$ cycles, we replace $v$ by $n$ blue vertices $v_1, v_2, \ldots , v_n$ and set $v_i$ adjacent to the vertices that $v$ was adjacent to in the $i$-th cycle.  Then, we include a red vertex $v_0$ which is connected to $v_1, v_2, \ldots, v_n$.  An example is provided in Figure \ref{fig:connected-tverberg}.  We see that $G'$ is a connected cactus graph.  Moreover, in $G'$ no two cycles share a vertex.  

\begin{figure}
\centering
\includegraphics[scale=0.7]{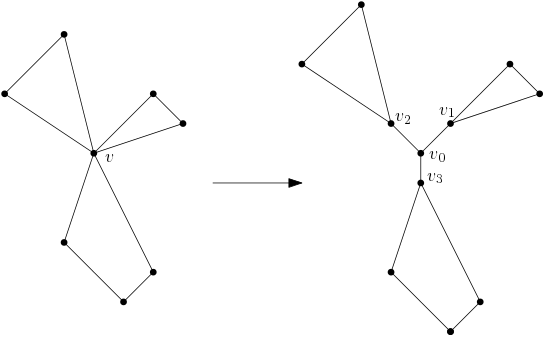}
\caption{The construction of $G'$ from $G$.}
\label{fig:connected-tverberg}
\end{figure}

Let us show that no red vertex of $G'$ is in a cycle.  Suppose towards contradiction that $v_0$ is a red vertex which is in a cycle $L_r$.  Notice that in the cycle $v_0$ must be adjacent to some $v_i$, which in turn is in a blue cycle $L_b$.  If we follow $L_r$ in the direction from $v_i$ to $v_0$, at some point it must return to $L_b$.  However, it cannot return to $L_b$ through $v_i$, so the first time it gets back to $L_b$ it must be using a different vertex $x$.  Notice that this induces a path in $G$ from $v$ to $x$ which does not use vertices of $L_b$, and contradicts the fact that $G$ is a cactus graph.  Hence, no red vertex of $G'$ can lie in a cycle.

Along with $G'$, we define a copy $U'$ of $U$.  If $v$ was a vertex of $U$ in at most one cycle of $G$, we set its corresponding copy in $G'$ to be part of $U'$.  If $v$ was a member of $U$ in more than one cycle, we set only $v_0$ (the red copy) to be part of $U'$.  If we can find two vertices in $G'$ such that removing them leaves only connected components with at most $2k-1$ points of $U'$ each, the corresponding pair (or single vertex) in $G$ would also work, as the only change could be that $G$ would be shattered into more connected components.

Now we construct a second auxiliary graph $G''$ from $G'$.  For each cycle in $G'$, include a vertex in $G''$, and for every vertex in $G'$ which is contained in no cycle, include a vertex in $G''$.  Set $u$, $v$ in $G''$ to be adjacent if there is an edge connecting some vertex corresponding to $u$ and some vertex corresponding to $v$.  Since $G'$ is connected, so is $G''$.  Let us show that $G''$ has no cycles.

If there were a cycle $L'$ in $G''$, all the vertices of $L'$ would correspond to cycles in $G'$.  Take $u$ one of its vertices.  Using the cycle $L'$, we can find a cycle in $G'$ different from the one corresponding to $u$ which intersects it.  This contradicts the construction of $G'$, in which no two cycles share a vertex.

Thus, $G''$ is a tree.  Now to each vertex in $G''$ we assign the number of vertices in $U'$ that correspond to it.  Following the same method as in the proof of Theorem \ref{theorem-for-trees}, we can find a vertex $p$ in $G''$ the removal of which leaves connected components with at most $2k-1$ vertices of $U'$ each.  If $p$ corresponds to a vertex in $G'$, removing it (and any other point $y$) gives us the partition we seek.

If $p$ corresponds to some cycle $L$ in $G'$, we have to do more work.  For this, we can represent the vertices of $L$ as the vertices of a regular $|L|$-gon in the plane in the order of the cycle.  To each vertex $v$ of $L$ we assign two numbers.  The first, $\mu_1(v)$ is $0$ or $1$ depending on whether $v$ is in $U'$ or not.  The second, $\mu_2(v)$, is computed as follows.  Notice that when we removed $p$ from $G''$, we had several connected components left.  Each of these corresponds to a connected component in $G'$ which is connected to exactly one point of $L$.  We take $\mu_2(v)$ to be the number of points of $U'$ contained in connected components of $G'\setminus L$ which were connected to $v$ in $L$.

Notice that $\mu_1$ and $\mu_2$ are two discrete finite measures in the plane whose sum is $4k-2$.  By the discrete version of the ham sandwich theorem, there is a line $\ell$ which leaves at most half of each measure in each halfspace.  We may assume without loss of generality that $\ell$ contains exactly two points $x,y$ of $L$.  Let us show that this choice of $x,y$ satisfies the conditions we seek.

We let $k_1,k_2$ be the total measures of $\mu_1,\mu_2$, respectively, so that $k_1+k_2=4k-2$.  When we remove the points $x,y$, we are left with at most two connected components which contain points of $L$.  By the construction of $\ell$, these have no more than $\lfloor k_1/2\rfloor+\lfloor k_2 /2 \rfloor$ points of $U'$ each.  We also have the connected components which were connected only to $x$ or only to $y$, but these each have at most $2k-1$ points of $U'$ by the construction of $L$.  Thus, removing $x,y$ gives us the partition we want, since each part contains at most $\max\{\lfloor k_1/2\rfloor+\lfloor k_2 /2 \rfloor, 2k-1\}=2k-1$.

Notice as well that we may even allow for $U$ to have repeated vertices, and the only change would be in allowing the vertices of $G'$ to preserve such multiplicities in $U'$, and in the definition of $\mu_2$.

Let us now consider the time required to find vertices $x,y$ of the form desired. For each pair of vertices $(v,w)$ of $G$, let $u(v,w)$ be the tuple corresponding to the numbers of elements of $U$ in the components of $G\setminus \{v,w\}$.  We claim that for each fixed $v$, identifying all $u(v,w)$, as $w$ varies, takes time at most linear in $n$.

Let $G_v = G\setminus \{v\}$. Observe that for each vertex $w$ of $G_v$, the tuple $u(v,w)$ is trivial if $w$ is in a cycle within $G_v$, since deleting $w$ gives only one connected component.  We therefore proceed by collapsing all cycles in $G_v$ to single vertices, where all edges incident to vertices of the cycle are incident to the new \emph{supervertex}.  The set $U$ becomes a multiset in which each supervertex may occur multiple times, according to the membership in $U$ of the vertices that were collapsed.

In order to collapse cycles in $G_v$, we begin by finding a spanning tree $T$ of $G_v$, which takes linear time. For each edge of $G_v$ not in $T$, we identify the path within $T$ between endpoints of the edge; this takes time linear in the number of vertices along this path.  We then collapse the path to a single vertex. Since each edge in $G_v$ is collapsed at most once, we can collapse all cycles in this way in time linear in the number of edges of $G_v$.  It is a classic result that $G$ has $O(n)$ edges since it is a cactus graph; therefore, this time is linear in $n$.

After collapsing cycles, $G_v$ has been transformed into a tree, and we can proceed as in the proof of Theorem \ref{theorem-for-trees} to identify $u(v,w)$ for all remaining vertices $w$. Thus, letting $v$ vary, we are able to find $u(v,w)$ for all pairs $(v,w)$ in $O(n^2)$ time, as desired.
\end{proof}






\section{Hardness of general Radon partitions}
\label{sec:geodetic-negative}

\begin{proof}[Proof of Theorem \ref{thm:onepartition}.]
Given a Boolean formula $\Phi$ in conjunctive normal form, we will describe a graph $G=(V,E)$ and $W\subset V$ with size $\poly(n)$, where a Radon partition $(W^+,W^-)$ corresponds to a truth assignment satisfying $\Phi$.  This suffices to prove the theorem, as SAT is a known NP-hard problem, while counting the number of satisfying assignments is known to be \#P-hard.

\textbf{Claim.} It suffices to consider $\Phi$ where every clause contains either all positive literals or all negative literals.

To prove the claim, suppose that $\Phi$ is given in terms of variables $x_i$.  Replace every occurrence of $-x_i$ in $\Phi$ by a new variable $y_i$, and add the clauses $(x_i\vee y_i)$ and $(-x_i\vee -y_i)$ for every $i$, which ensures that $x_i$ and $y_i$ have opposite truth value.  In the new formula $\Phi'$, every clause must have either all positive literals or all negative literals; and solutions to $\Phi'$ exactly correspond to solutions to $\Phi$, as desired.

Suppose now that $\Phi$ takes the form given in the claim, with variables $x_1,\ldots,x_n$ occurring as positive literals in clauses $C^+_1,\ldots,C^+_\ell$, and occurring as negative literals in clauses $C^-_1,\ldots,C^-_m$.  For convenience, we will also suppose that $\ell,m\ge 1$; if this is not true, we may simply add a new variable and a corresponding new clause containing it.

We now outline the construction of our graph $G$ (shown in Figure \ref{fig:SAT-graph}).  The vertex set $V$ will consist of the following parts:
\begin{itemize}
\item the set $W^0=\{w_1,\ldots,w_n\}$, corresponding to the variables $x_1,\ldots,x_n$,
\item vertices $w^+$ and $w^-$, which will define which variables are True and False, respectively,
\item the set $V^+$, which will allow us to calculate which of the clauses $C^+_j$ are True,
\item the set $V^-$, which will allow us to calculate which clauses $C^-_k$ are True.
\end{itemize}

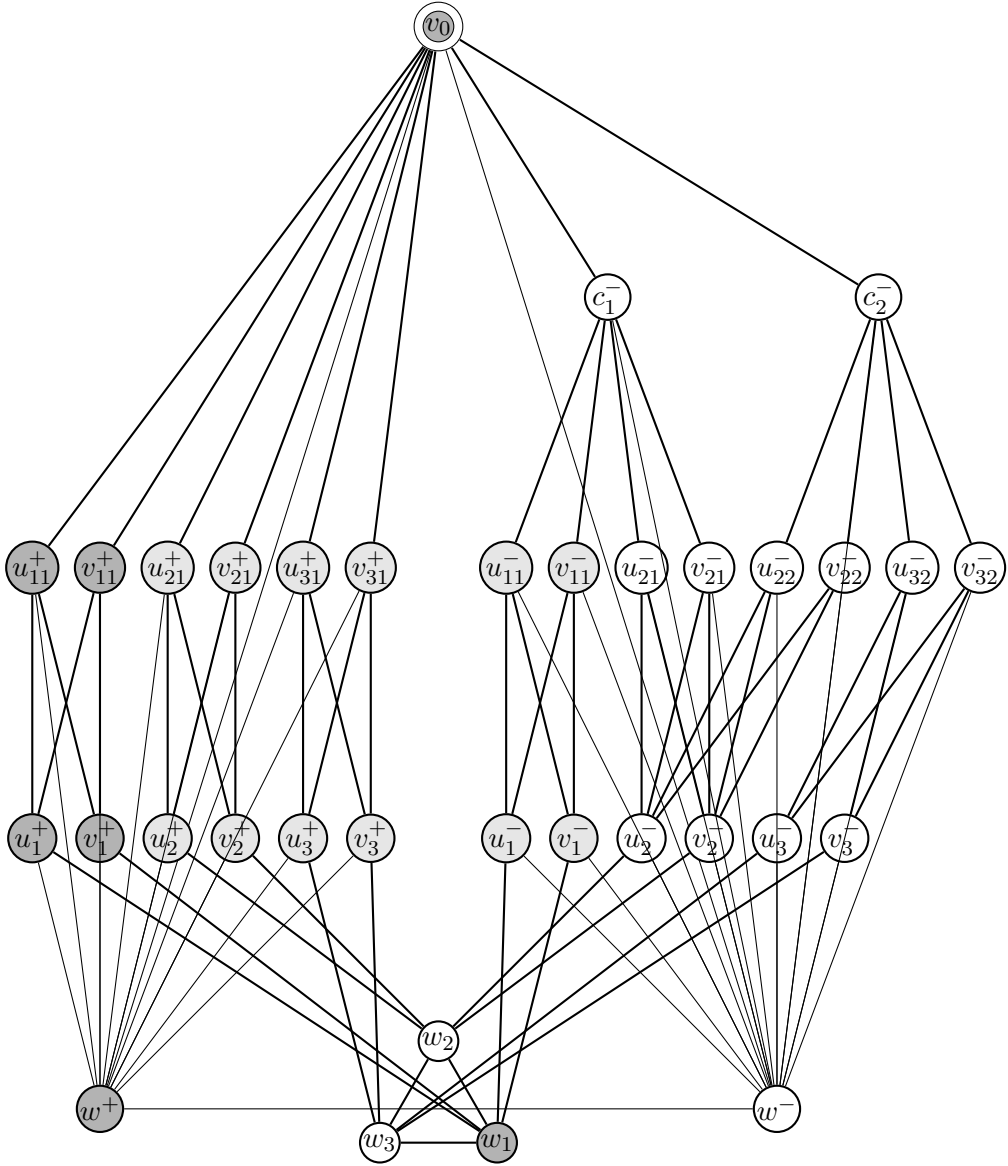
\begin{figure}
\begin{tikzpicture}[scale=0.9]
\tikzstyle{vertex}=[circle,draw,thick=black,fill=black!10,minimum size=15pt,inner sep=0pt]
\tikzstyle{vertex_light}=[circle,draw,thick=black,fill=black!0,minimum size=15pt,inner sep=0pt]
\tikzstyle{vertex_dark}=[circle,draw,thick=black,fill=black!30,minimum size=15pt,inner sep=0pt]
\tikzstyle{vertex_double}=[circle,draw,double=white,double distance=3 pt,fill=black!30,minimum size=15pt,inner sep=0pt, outer sep = 2pt]
\tikzstyle{edge} = [draw,thick,-]
\tikzstyle{dotted} = [draw,thin,-]
\node[vertex_dark] (w^+) at (-5,0) {$w^+$};
\node[vertex_light] (w^-) at (5,0) {$w^-$};
\node[vertex_double] (v_0) at (0,16) {$v_0$};
\path[dotted] (v_0) -- (w^+);
\path[dotted] (v_0) -- (w^-);
\path[dotted] (w^+) -- (w^-);
\foreach \pos/\name/\option in {{(330:1)/w_1/vertex_dark}, {(90:1)/w_2/vertex_light}, {(210:1)/w_3/vertex_light}}
    \node[\option] (\name) at \pos {$\name$};
\foreach \pos/\name/\option in {{(-6,4)/u_1^+/vertex_dark}, {(-5,4)/v_1^+/vertex_dark}, {(-4,4)/u_2^+/vertex}, {(-3,4)/v_2^+/vertex}, {(-2,4)/u_3^+/vertex}, {(-1,4)/v_3^+/vertex}, {(-6,8)/u_{11}^+/vertex_dark}, {(-5,8)/v_{11}^+/vertex_dark}, {(-4,8)/u_{21}^+/vertex}, {(-3,8)/v_{21}^+/vertex}, {(-2,8)/u_{31}^+/vertex}, {(-1,8)/v_{31}^+/vertex}}
    {\node[\option] (\name) at \pos {$\name$};
    \path[dotted] (\name) -- (w^+);}
\foreach \pos/\name/\option in {{(1,4)/u_1^-/vertex}, {(2,4)/v_1^-/vertex}, {(3,4)/u_2^-/vertex_light}, {(4,4)/v_2^-/vertex_light}, {(5,4)/u_3^-/vertex_light}, {(6,4)/v_3^-/vertex_light}, {(1,8)/u_{11}^-/vertex}, {(2,8)/v_{11}^-/vertex}, {(3,8)/u_{21}^-/vertex_light}, {(4,8)/v_{21}^-/vertex_light}, {(5,8)/u_{22}^-/vertex_light}, {(6,8)/v_{22}^-/vertex_light}, {(7,8)/u_{32}^-/vertex_light}, {(8,8)/v_{32}^-/vertex_light}, {(2.5,12)/c_1^-/vertex_light}, {(6.5,12)/c_2^-/vertex_light}}
    {\node[\option] (\name) at \pos {$\name$};
    \path[dotted] (\name) -- (w^-);}
\foreach \source/ \dest in {w_1/w_2, w_1/w_3, w_2/w_3, u_1^+/w_1, v_1^+/w_1, u_2^+/w_2, v_2^+/w_2, u_3^+/w_3, v_3^+/w_3,
    u_1^-/w_1, v_1^-/w_1, u_2^-/w_2, v_2^-/w_2, u_3^-/w_3, v_3^-/w_3, u_1^+/u_{11}^+, v_1^+/u_{11}^+, u_1^+/v_{11}^+, v_1^+/v_{11}^+, u_2^+/u_{21}^+, v_2^+/u_{21}^+, u_2^+/v_{21}^+, v_2^+/v_{21}^+, u_3^+/u_{31}^+, v_3^+/u_{31}^+, u_3^+/v_{31}^+, v_3^+/v_{31}^+, u_1^-/u_{11}^-, v_1^-/u_{11}^-, u_1^-/v_{11}^-, v_1^-/v_{11}^-, u_2^-/u_{21}^-, v_2^-/u_{21}^-, u_2^-/v_{21}^-, v_2^-/v_{21}^-, u_2^-/u_{22}^-, v_2^-/u_{22}^-, u_2^-/v_{22}^-, v_2^-/v_{22}^-, u_3^-/u_{32}^-, v_3^-/u_{32}^-, u_3^-/v_{32}^-, v_3^-/v_{32}^-, u_{11}^+/v_0, v_{11}^+/v_0, u_{21}^+/v_0, v_{21}^+/v_0, u_{31}^+/v_0, v_{31}^+/v_0, u_{11}^-/c_1^-, v_{11}^-/c_1^-, u_{21}^-/c_1^-, v_{21}^-/c_1^-, u_{22}^-/c_2^-, v_{22}^-/c_2^-, u_{32}^-/c_2^-, v_{32}^-/c_2^-, c_1^-/v_0, c_2^-/v_0}
    \path[edge] (\source) -- (\dest);
\end{tikzpicture}
\caption{The graph $G$ derived from the Boolean formula $\Phi = C_1^+\wedge C_1^-\wedge C_2^-$, with 
$C_1^+ = x_1\vee x_2\vee x_3, C_1^- = -x_1 \vee -x_2, C_2^- = -x_2 \vee -x_3$. Note that the vertex $v_0$ could also be labeled $c_1^+$, $a_1^+$, and $a_2^-$. Darkly shaded vertices are in the convex hull of $\{w^+,w_1\}$, lightly shaded vertices are in the convex hull of $\{w^-,w_2,w_3\}$, and moderately shaded vertices are in neither. Note that vertex $v_0$ is in both convex hulls, indicating that the truth assignment $\{x_1=\text{True},\,x_2=x_3=\text{False}\}$ is a solution to $\Phi$.}
\label{fig:SAT-graph}
\end{figure}

The sets $V^+$, $V^-$ overlap in the vertex $v_0$, which will allow us to test if all clauses $C_j$ and $D_k$ are satisfied.

The input $W$ in the graph-Radon problem will be defined as $W^0\cup \{w^+,w^-\}$.  The motivation is that in any $2$-partition of $W$ such that $w^+$ and $w^-$ are in different parts, we consider the variable $x_i$ to be True if $w_i$ is in the same part as $w^+$, and to be False if $w_i$ is in the same part as $w^-$.  We will eventually deal with the case where $w^+$ and $w^-$ are in the same part; however, for the moment let us assume that each of the variables is designated unambiguously True or False.

Each vertex $v$ in $V$, with the exception of $v_0$, possesses a parameter called \emph{height}, denoted $h(v)$.  The height of a vertex is a nonnegative integer, equal to $0$ for $v\in W^0\cup \{w^+,w^-\}$ and taking on various positive values for $v\in V^+$ or $v\in V^-$.  The importance of the height is as follows: Suppose we start out with a subset of the height-0 vertices and wish to find the convex hull of these vertices.  In each \emph{extension step}, we add all vertices which lie on shortest paths between existing vertices.  Then, we will design $G$ such that in the first extension step we acquire only height-1 vertices, in the second extension step only height-2 vertices, and so on.  We will never end up with more height-0 vertices than we started with, or with more height-$t$ vertices than we acquired at extension step $t$.  This means that we will be able to analyze the process of building the convex hull, one extension step at a time.

In addition to this property, we will construct $G$ with the following property: The convex hull of a subset of $W$ contains a vertex of $V^+\setminus \{v_0\}$ only if it contains $w^+$; likewise it contains a vertex of $V^-\setminus \{v_0\}$ only if it contains $w^-$.  Intuitively, $V^+$ is the part of the graph that handles True variables, and $V^-$ is the part of the graph that handles False variables.  We will now present the structure of $V^+$ and $V^-$.  Given a partition of $W$ in which $w^+$ and $w^-$ lie in different parts, we will refer to the convex hull of the part with $w^+$ as the \emph{positive convex hull}, and likewise define the \emph{negative convex hull}.

At height 1 in $G$, we define vertices corresponding to two copies of each literal:
$$u_1^+,v_1^+,u_2^+,v_2^+,\ldots,u_n^+,v_n^+\in V^+\hskip .2 in \text{and} \hskip .2 in u_1^-,v_1^-,u_2^-,v_2^-,\ldots,u_n^-,v_n^-\in V^-.$$
The vertices $u_i^+,v_i^+$ are both included in the positive convex hull if (and only if) $x_i$ is True, and $u_i^-,v_i^-$ are both included in the negative convex hull if (and only if) $x_i$ is False.  To achieve this, we simply make $u_i^+,v_i^+$ adjacent to both $w_i$ and $w^+$, while $u_i^-,v_i^-$ are adjacent to both $w_i$ and $w^-$.  Thus, the shortest path between $w_i$ and $w^+$ includes $u_i^+,v_i^+$, justifying our claim that height-1 vertices are included in the convex hull on the first extension step.

The construction of $V^+$ and $V^-$ will be exactly symmetric from this point, and we will therefore present only the construction of $V^+$.  The key here is that every vertex of $V^+$ is adjacent to $w^+$.  The reason for this is that every shortest path between vertices of $V^+$ thus has length 1 or 2.  It will be easy to ensure, therefore, that no pairs of vertices augment the convex hull in undesirable ways.

At height 2, we define vertices $u_{ij}^+,v_{ij}^+\in V^+$ for each occurrence of literal $x_i$ in clause $C^+_j$. (Likewise, we define $u_{ik}^-,v_{ik}^-\in V^-$ for each occurrence of literal $-x_i$ in clause $C^-_k$.)  We let $u_{ij}^+,v_{ij}^+$ be adjacent to $u_i^+$ and $v_i^+$. Since $u_{ij}^+,v_{ij}^+$ lie in $V^+$, they are also adjacent to $w^+$.  Then, $u_{ij}^+,v_{ij}^+$ lie on shortest paths between $u_i^+$ and $v_i^+$; thus $u_{ij}^+,v_{ij}^+$ are included in the positive convex hull at the second extension step if and only if $x_i$ is True.

At height 3, we define a vertex $c_j^+\in V^+$ for each clause $C^+_j$ such that $c_j^+$ is adjacent to $u_{ij}^+,v_{ij}^+$ for every value of $i$ such that $x_i$ occurs in $C^+_j$. Then, $c_j^+$ lies on a shortest path between $u_{ij}^+$ and $v_{ij}^+$; thus $c_j^+$ is included in the positive convex hull at the third extension step if and only if $x_i$ is True.  In essence, we have defined an OR gate for the positive literals that appear in the clause $C^+_j$.

It may seem obscure why we defined $u_{ij}^+,v_{ij}^+$, instead of letting $c_j^+$ be adjacent to $u_i^+,v_i^+$ directly.  The reason is that we must ensure that vertices of greater height do not cause vertices of lower height to be included in the convex hull.  If $c_j^+,d_j^+$ were adjacent to $u_i^+,v_i^+$, then any literal $x_i$ common to clauses $C^+_j$ and $C^+_{j'}$ would be switched to True whenever $c_j^+,c_{j'}^+$ were both included in the positive convex hull.  This is undesirable; therefore, we introduce $u_{ij}^+,v_{ij}^+$, which provide additional space and prevent $u_i^+$ from being included in shortest paths from $c_j^+$ and $c_{j'}^+$.  The reason for this is that any shortest path between $c_j^+$ and $c_{j'}^+$ must have length 2, since both vertices are adjacent to $w^+$.

Recall that we want to create a correspondence between solutions to $\Phi$ and Radon partitions of the vertices $W$.  Given such a partition, corresponding to an assignment of True/False to every variable $x_i$, we have seen that each $c_j^+$ lies in the positive convex hull if and only if $C^+_j$ is satisfied. We will now define an AND gate to determine when all of the clauses $c_j^+$ lie in the positive convex hull.

For each $j$ with $1\le j\le \ell$, recursively define vertices $a_j^+\in V^+$ at height $2+j$ as follows: Let $a_1^+=c_1^+$, and for $j\ge 1$ let $a_j^+$ be adjacent to $a_{j-1}^+$ and $c_j^+$.  Observe that $a_j^+$ lies in the positive convex hull if and only if the vertices $c_1^+,c_2^+,\ldots,c_j^+$ lie in the positive convex hull.  Likewise, we define vertices $a_k^-$ for $1\le k\le m$.  The crucial detail is that we identify $a_\ell^+$ and $a_m^-$ as a common vertex $v_0$.  (This is why the height of $v_0$ is not defined, since $2+\ell$ may not be equal to $2+m$.)  Observe that for a partition of $W$, the vertex $v_0$ lies in both the positive and negative convex hulls exactly when every $C_j^+$ and every $C_k^-$ is satisfied.  Moreover, $v_0$ is the only possible overlap between the positive and negative convex hulls. Hence, Radon partitions of $W$ correspond exactly to truth assignments that satisfy $\Phi$.

One final detail remains: We assumed that, in any partition of $W$, the vertices $w^+$ and $w^-$ lie in different parts.  In order to enforce this requirement, we let $W^0$ form a clique, and add an edge between $w^+$ and $w^-$.  For any partition of $W$ in which $w^+$ and $w^-$ lie in the same part, one of the two parts must contain only elements of $W^0$, and therefore forms a clique, which is its own convex hull.  We have defined the graph $G$ in such a way that the convex hull of $W^0\cup\{w^+,w^-\}\backslash W'$ does not overlap $W'$, for any subset $W'\subseteq W^0$.  (The reason for the edge between $w^+$ and $w^-$ is so that shortest paths between vertices of $V^+$ and $V^-$ lie through $w^+$ and $w^-$ and therefore do not intersect $W'$.)  We conclude that, in any Radon partition, the vertices $w^+$ and $w^-$ must indeed lie in different parts.  This completes our proof.
\end{proof}

Our approach suggests an alternate proof for the statement (proven in \cite{coelho2015}) that it is NP-hard to calculate the geodetic Radon number.

\begin{proof}[Proof of hardness of geodetic Radon number]
Suppose that we are given a CNF formula $\Phi$ with variables $x_1,\ldots,x_n$ and clauses $C_1,\ldots,C_m$.  Let us define a formula $\Phi^\ell$ on variables $\cup_{j=1}^\ell \{x_1^j,\ldots,x_n^j\}$.  Let the clauses take the form $C^1_{i_1}\vee C^2_{i_2}\vee \cdots \vee C^\ell_{i_\ell}$ for all possible selections $i_1,\ldots,i_\ell\in \{1,\ldots,m\}$, where $C^j_{i_j}$ denotes the clause $C_{i_j}$ written using the variables $x^j_i$.  Thus, each clause of $\Phi^\ell$ corresponds to an $\ell$-tuple of clauses of $\Phi$, and is satisfied if any of the constituent clauses is satisfied.

\textbf{Claim 1.} $\Phi^\ell$ is satisfiable if and only if $\Phi$ is satisfiable.

If $\Phi$ has a satisfying assignment $x_i=t_i$, then we can create a satisfying assignment for $\Phi^\ell$ by setting $x_i^j=t_i$ for each $j$.  Conversely, suppose that $\Phi^\ell$ has a satisfying assignment $x_i^j=t_i^j$.  We claim that for some $j$, the assignment $x_i=t_i^j$ is satisfying for $\Phi$.  Suppose towards contradiction that for each $j$, there exists a clause $C_{i_j}^j$ of $\Phi$ which is not satisfied by the assignment $x_i=t_i^j$.  Then, the clause $C^1_{i_1}\vee C^2_{i_2}\vee \cdots \vee C^\ell_{i_\ell}$ of $\Phi^\ell$ is not satisfied by $x_i^j=t_i^j$, giving us our contradiction.

\textbf{Claim 2.} Suppose we remove up to $\ell-1$ variables from $\Phi^\ell$ (i.e., remove all occurrences of these variables without deleting clauses) to obtain $\tilde{\Phi}$. Then, $\tilde{\Phi}$ is satisfiable if $\Phi$ is satisfiable.

After removing up to $\ell-1$ variables, there must remain some value $j_0$ such that none of $x_1^{j_0},x_2^{j_0},\ldots,x_n^{j_0}$ is removed.  If $x_i=t_i$ is a satisfying assignment for $\Phi$, then assigning $x_i^{j_0}=t_i$ (and assigning all other $x_i^j$ arbitrarily) gives us a satisfying assignment for $\Phi^\ell$, as desired.

Let $G(\Phi)$ denote the graph defined in the proof of Theorem \ref{thm:onepartition}, where $\Phi$ is a CNF formula.   Now, consider the graph $G=G(\Phi^\ell)$, where we will pick $\ell$ to be a fixed constant, so that the size of $G$ is polynomial in $n$. Let $V^+,V^-,W^0,w^+,w^-$ be as in the preceding proof. Consider the problem of determining whether $G$ has Radon number $\ell n+2$.  Let $W$ be a subset of the vertices of $G$ such that $|W|=\ell n+2$.  Note that if $W=W^0\cup \{w^+,w^-\}$, then there exists a Radon partition of $W$ if and only if $\Phi$ is satisfiable (by Claim 1 and Theorem \ref{thm:onepartition}).  Hence, the Radon number of $G$ is greater than $\ell n+2$ if $\Phi$ is not satisfiable.  We will show now that the Radon number is at most $\ell n+2$ if $\Phi$ \emph{is} satisfiable.  We divide into two cases based upon the structure of $W$.

\textbf{Case 1.} Both $|W\cap V^+|$ and $|W\cap V^-|$ are at most $\lfloor (\ell-1)/2\rfloor$.

This condition implies that $W$ differs from $W^0\cup \{w^+,w^-\}$ in at most $\ell-1$ vertices.  Let us modify $\Phi^\ell$ to $\tilde{\Phi}$ by removing the variables corresponding to these vertices.  Then, Claim 2 implies that if $\Phi$ is satisfiable, then $\tilde{\Phi}$ is also; hence there exists a Radon partition of $W$.

\textbf{Case 2.} $|W\cap V^+|$ or $|W\cap V^-|$ is greater than $\lfloor (\ell-1)/2\rfloor$.

Without loss of generality, suppose $|W\cap V^+|\ge \lfloor (\ell-1)/2\rfloor$.  Note that, by construction, there are no triangles within $V^+$; therefore, by choosing $\ell$ large enough, we can ensure that there exist vertices $y_1,y_2,z_1,z_2\in W\cap V^+$ such that $y_1,y_2$ and $z_1,z_2$ are non-adjacent.  Then, $w^+$ is included on shortest paths between $y_1,y_2$ and between $z_1,z_2$. Hence, any partition of $W$ is a Radon partition if $y_1,y_2$ are in different parts, as are $z_1,z_2$.

We conclude that determining the Radon number of $G$ allows us to infer the satisfiability of $\Phi$, from which the result follows.
\end{proof}


\section*{Acknowledgments}

The authors would like to thank Jes\'us De Loera and Reuben La Haye for their helpful comments during this work.  D.R.~was supported by a National Science Foundation Graduate Research Fellowship under Grant No.~1122374.


\end{document}